\documentclass[a4paper,twoside,10pt,fleqn]{article}

\usepackage{graphicx}
\usepackage{latexsym}
\usepackage{amsmath}
\usepackage{amssymb}
\usepackage{amsthm}
\usepackage{color}
\usepackage{hyperref}

\newcommand{\versie}{2010/jul/30}
\newcommand{\kortetitel}{Pacer cell}
\newcommand{\korteauteur}{BBEGH}

\title{Pacer cell response to periodic Zeitgebers}

\author{D.G.M. Beersma, H.W. Broer, K. Efstathiou, K.A. Gargar and I. Hoveijn\\University of Groningen\\Department of Chronobiology, PO Box 14, 9750 AA Haren\\Johann Bernoulli Institute for Mathematics and Computer Science, PO Box 407\\The Netherlands}

\date{\versie}

\newcommand{\abc}{
\renewcommand{\theenumi}{\alph{enumi}}
\renewcommand{\labelenumi}{\theenumi)}
\itemsep 0pt}

\newcommand{\iii}{
\renewcommand{\theenumi}{\roman{enumi}}
\renewcommand{\labelenumi}{\theenumi)}
\itemsep 0pt}

\theoremstyle{plain}
\newtheorem{theorem}{Theorem}
\newtheorem{proposition}[theorem]{Proposition}
\newtheorem{lemma}[theorem]{Lemma}

\newtheorem{corollary}[theorem]{Corollary}
\newtheorem{definition}[theorem]{Definition}

\newtheorem{Remark}{Remark}

\newenvironment{remark}{\begin{Remark} \begin{rm}}{\hfill$\blacktriangleright$\end{rm} \end{Remark}}

\newcommand{\commentaar}[1]{{}}
\newcommand{\opmerking}[1]{\texttt{(#1)}}
\newcommand{\avg}[1]{<\! #1 \!>}

\newcommand{\cO}{\mathcal{O}}
\newcommand{\cT}{\mathcal{T}}
\newcommand{\cinf}{C^{\infty}}
\newcommand{\eps}{\varepsilon}

\newcommand{\fR}{\mathbb{R}}

\newcommand{\zgbr}{Z}
\newcommand{\fmu}{f_{\mu}}
\newcommand{\Fmu}{F_{\mu}}

\begin{document}
\maketitle


\commentaar{
- Arnold map versus Arnold family
+ 'influence' naar 'stimulus' overal
}

\begin{abstract}\noindent
Almost all organisms show some kind of time periodicity in their behavior. Especially in mammals the neurons of the suprachiasmatic nucleus form a biological clock regulating the activity-inactivity cycle of the animal. This clock is stimulated by the natural 24-hour light-dark cycle. In our model of this system we consider each neuron as a so called phase oscillator, coupled to other neurons for which the light-dark cycle is a Zeitgeber. To simplify the model we first take an externally stimulated single phase oscillator. The first part of the phase interval is called the active state and the remaining part is the inactive state. Without external stimulus the oscillator oscillates with its intrinsic period. An external stimulus, be it from activity of neighboring cells or the periodic daylight cycle, acts twofold, it may delay the change form active to inactive and it may advance the return to the active state. The amount of delay and advance depends on the strength of the stimulus. We use a circle map as a mathematical model for this system. This map depends on several parameters, among which the intrinsic period and phase delay and advance. In parameter space we find Arnol'd tongues where the system is in resonance with the Zeitgeber. Thus already in this simplified system we find entrainment and synchronization. Also some other phenomena from biological experiments and observations can be related to the dynamical behavior of the circle map.
\end{abstract}

\begin{center}
\end{center}

\section{Introduction}\label{sec:intro}
\commentaar{
- ritmiek, circadische klok, Zeitgeber
- collectie van pacer cellen, 2 fenomenen, collectief gedrag en reactie op Zeitgeber in de vorm van 'entrainment' of 'synchronization'.
- enkele pacer cell vertoont op lager niveau ook 'entr' en 'sync', globale beschrijving model, lit refs 'L. Glass' en meer van dat soort, rol van de Zeitgeber is op dit niveau tweeledig, keuzen voor Zeitgeber en waarom
- wat vinden we en wat is de betekenis daarvan
}

\subsection*{Setting of the problem}\label{sec:introset}
Rhythmic behavior is present in almost all organisms. Their rhythms can be autonomous, but are more often externally stimulated. One such stimulus is the 24-hour natural light-dark cycle which governs the activity-inactivity cycle of many animals and plants. The latter is the most common \emph{Zeitgeber} or periodic stimulus, although an alternating high-low temperature cycle is another example of a Zeitgeber. Shaped by several millions of years of evolution under this 24-hour light-dark cycle \cite{j05}, many organisms still exhibit these behavioral rhythms (often with slightly different period) even in conditions without information on the alternation of light and darkness. How the almost-24-hour \emph{intrinsic period} of the internal rhythm synchronizes with or entrains to an external Zeitgeber is one of the major questions in circadian biology, see \cite{asc81,da01}.

In mammals, the circadian clock resides in the suprachiasmatic nucleus, a neuronal hypothalamic tissue residing just above the optic chiasm. It consists of about 10\,000 interconnected neurons or pacer cells \cite{hh,me,moo,rfdm}. Experimental evidence seems to support the model of the suprachiasmatic nucleus as a collection of so called phase oscillators, which was put forward in 1980 to explain circadian rhythms \cite{ert}. Each oscillator has its own intrinsic period of activity-inactivity \cite{wlmr} and interactions among them are believed to synchronize their oscillations \cite{hsknh,hnsh,shkoh}. This model is used to explain the observation that the activity-inactivity cycle of an organism closely follows the period of a Zeitgeber like the 24-hour light-dark cycle.

However, before trying to analyze a model for the collection of pacer cells, we consider a model for a single cell. A collection of interacting pacer cells may mathematically be modelled by a (large) number of coupled almost identical oscillators. Such a model is rather involved, therefore we first make the following simplifying assumption. In our model each oscillator experiences the external (averaged) forcing by the collection of oscillators but does not influence the dynamics of the others. Thus the coupling we consider is asymmetric. This approach is also widely used in astro-dynamics \commentaar{refs!}. The basis for our model of a single pacer cell is given in \cite{bbhd, ert} where the state of a pacer cell is determined by the phase in its activity-inactivity cycle. The stimulus of an external Zeitgeber may advance or delay the phase depending on the phase itself but also on parameters characterizing the pacer cell. The parameters essentially consist of the intrinsic period of the pacer cell, the intrinsic length of the activity interval and the strength of the interaction with an external stimulus. The latter can be the relative number of active pacer cells in the environment as well as a stimulus originating from an external light, temperature, etc level. In our model the Zeitgeber will be a periodic, quasi-periodic or even more general function of time. But in this article we restrict to a periodic Zeitgeber where the period is an average or prevailing period in the natural light and dark cycle or in the collective behavior of pacer cells in the environment.

Models based on phase oscillators date back to at least 1967, see \cite{win67}, and have been studied by several others \cite{blmch,ghrkh,okz}. For an overview see \cite{gla}.

\subsection*{Main questions}\label{sec:intromq}
Thus we study a model of a single pacer cell, stimulated by its environment but not contributing to the collective behavior. For such a situation the main questions we wish to address are
\begin{enumerate}\itemsep 0pt\parsep 0pt
\item Can a single pacer cell synchronize with or entrain to a periodic Zeitgeber?
\item If so, how does this depend on properties of the pacer cell?
\end{enumerate}
The mathematical model we use to describe a pacer cell stimulated by a periodic Zeitgeber, is a dynamical system, more specifically a map on the circle. This map depends on the parameters characterizing the pacer cell. Typical dynamics of a circle map most relevant in view of the questions above is dynamics of fixed points and dynamics of periodic points. Fixed points correspond to \emph{entrainment} of the pacer cell which means that the sequence of onset times of the activity interval has the same period as the Zeitgeber. Periodic points correspond to \emph{synchronization} meaning that there are an integer number of $p$ different onset times of the activity interval during another integer number of $q$ periods of the Zeitgeber. Such points are called $p:q$ periodic points. In this vocabulary fixed points are $1:1$ periodic points, in other words entrainment is a special kind of synchronization. Then the questions for the single pacer cell are translated into the following questions for the mathematical model.
\begin{enumerate}\itemsep 0pt\parsep 0pt
\item Do stable fixed or periodic points exist for the map on the circle?
\item If so, for which domain in parameter space?
\end{enumerate}
%

\subsection*{Summary of results}\label{sec:introres}
The analysis of the model for a single pacer cell shows that both entrainment and synchronization are possible. For weak interaction with the Zeitgeber the parameter space is divided in regions with unique and distinct dynamics. When we fix all parameters except the intrinsic period of the cell and the strength of the interaction we get the so called Arnol'd tongues, see figure \ref{fig:arnoldtongues}. These are labelled by $p$ and $q$ such that for parameter values in the $p:q$ tongue, the dynamics of the map is $p:q$-periodic corresponding to synchronization of the pacer cell. This shows that for a fixed strength of interaction with the Zeitgeber there are \emph{ranges of synchronization}. Which means that for only a finite interval of values of the intrinsic period synchronization is possible. A phenomenon which is confirmed by biological experiments and observations, see \cite{asc65,asc81,ap78,gla,hod,mbrmgr,pd4,pd76,shkoh,str,uto,wev66,wev73}. As a remark we note that varying the intrinsic period at a fixed period of the Zeitgeber is equivalent to varying the period of the Zeitgeber while the intrinsic period is fixed. The largest range is the range of entrainment since the $1:1$ or main tongue is the largest. The location of the range of entrainment depends on the parameters governing phase advance and delay. In one extreme case only pacer cells with an intrinsic period smaller than the period of the Zeitgeber can be entrained, in the other extreme case only pacer cells with a larger intrinsic period can be entrained. Furthermore the onset time of the activity interval is an increasing function of the intrinsic period see section \ref{sec:matmodanl}, which has been observed in many organisms, see \cite{ap78,asc65,rm01,wev73}.

\section{Biological model}\label{sec:biomod}
In our biological model of pacer cells, the state of each pacer cell is determined by a single variable: a phase $\theta \in [0,\tau]$. This phase is the resultant of many biochemical processes. In an isolated pacer cell the phase increases in time with speed one until it reaches a value $\tau$, then it jumps to zero and starts to increase again. If the phase is between zero and a value $\alpha < \tau$ the cell is active and between $\alpha$ and $\tau$, the cell is inactive. Thus an isolated pacer cell shows a periodic activity-inactivity cycle with period $\tau$. We call $\alpha$ the length of the \emph{intrinsic activity interval} and $\tau$ is called \emph{intrinsic period} of the pacer cell. Both the values of $\alpha$ and $\tau$ are properties of the individual cell.

In a collection of pacer cells, considered as an organ without external stimulus, the interaction among the cells is modelled as follows. Let $t_n$ be the time at which the phase is zero. Then the cell is active during the time interval $t_n$ to $t_n + \alpha$. At the end of this interval an isolated cell would become inactive but now the cell remains active until a later time namely $t_n + \alpha + \eps \zgbr$, where $\zgbr$ is the fraction of active cells and $\eps$ is a small cell dependent parameter. We incorporate this \emph{phase delay} into the model by once delaying the phase by an amount of $\eps \zgbr$ at time $t_n + \alpha$, so instantly $\theta$ becomes $\theta-\eps \zgbr$. Similarly the inactivity period of the cell is shortened by activity of other cells. If at time $t_{n+1}$ the phase of the cell is below $\tau$ by an amount of $\eta \zgbr$ so that $\theta + \eta \zgbr = \tau$ the cell becomes active again. Thus there is a \emph{phase advance} of $\theta$ by an amount of $\eta \zgbr$ at time $t_{n+1}$, that is instantly $\theta$ becomes $\theta+\eta \zgbr$. Since the latter is equal to $\tau$, by previous assumptions the phase jumps to zero.

Before isolating a single pacer cell we note that in a collection of $n$ pacer cells the phase of cell $i$ is $\theta_i$. The quantity $\zgbr$ for cell $i$ depends on all phases $\theta_k$ except $\theta_i$ and possibly also on time $t$, so $\zgbr=\zgbr_i(\theta_1,\ldots,\hat{\theta}_i,\ldots,\theta_n,t)$, where the hat means $\theta_i$ excluded. However for reasons to be explained below, we will later on consider $\zgbr$ as a function of time only.

In a model with a large number of interacting pacer cells we assume that the interaction is uni-directional. That is the fraction of active cells stimulates a single cell, but the influence of a single cell on the collection is negligible. Therefore the quantity $\zgbr$ can be considered as only time dependent $\zgbr=\zgbr(t)$. Now our final model consists of a single pacer cell and we study its response to the external stimulus $\zgbr$. The pacer cell is characterized by four parameters $\eps$, $\eta$, $\alpha$ and $\tau$. We do not make a distinction whether $\zgbr$ is related to the fraction of active cells as a function of time or a daily light and dark signal. In our model $\zgbr$ is the \emph{Zeitgeber} for the cell and we will discuss several choices for it.

In principle the Zeitgeber $\zgbr$ can be any function of time. Here we will take it periodic, where the period is for example an average of the observed periods of the collection of pacer cells or an average period of an external stimulus like the daily light and dark cycle. But we may also take a much longer period to model seasonal effects. Since pacer cells may have different values of $\eps$, $\eta$, $\alpha$ and $\tau$ we are especially interested in the domain in parameter space where synchronization occurs.

\section{Mathematical model}\label{sec:matmod}
The system described in section \ref{sec:biomod} has a state which evolves in time. This calls for a \emph{dynamical systems} approach, as a general reference see \cite{bt}. In order to do so we have to identify a \emph{state space} and an \emph{evolution law}. This means that we should find one or more quantities describing the state of the system and then find an evolution law that uniquely defines a future state once an initial state is given. Here we first restrict to \emph{deterministic dynamical systems} that is we do not include noise or some other kind of random fluctuations. It will turn out that we can define a discrete dynamical system on the circle. Using this model we are able to answer some of the questions in section \ref{sec:intro} for a single pacer cell with a Zeitgeber.

\textbf{An isolated pacer cell.} The model we build is based on the phase $\theta$ of the cell. First we consider the cell as a `free running' oscillator, depending on the two parameters $\alpha$ and $\tau$. Later on we also include a periodic Zeitgeber with two additional parameters $\eps$ and $\eta$ governing the strength of the forcing.

Let $\theta :\fR \to [0,\tau)$ be the phase of the cell given by
\begin{equation*}
\theta(t) = 
\begin{cases}
t-t_n, & \text{for}\; t \in [t_n,t_{n+1})\\
0,     & \text{for}\; t = t_{n+1},
\end{cases}
\end{equation*}
where $\tau$ is the maximal value of $\theta$. When $\theta$ reaches the value $\alpha$, with $0 < \alpha < \tau$, the cell undergoes the transition from active to inactive. Suppose at $t = t_n$ the phase $\theta$ is zero. The phase increases with speed $1$ until it reaches the maximal value $\tau$, then $\theta$ instantly becomes zero again. This happens at time $t=t_{n+1} = t_n + \tau$. Thus the period of the `free running' oscillator is $\tau$. The transition times are $t_n = t_0 + n\,\tau$ and $t_n + \alpha$ where the cell state changes from inactive to active and from active to inactive respectively, see figure \ref{fig:theta}.

\begin{figure}[htbp]
\setlength{\unitlength}{1mm}
\begin{picture}(100,60)(0,-5)
\put(5,  5){\includegraphics{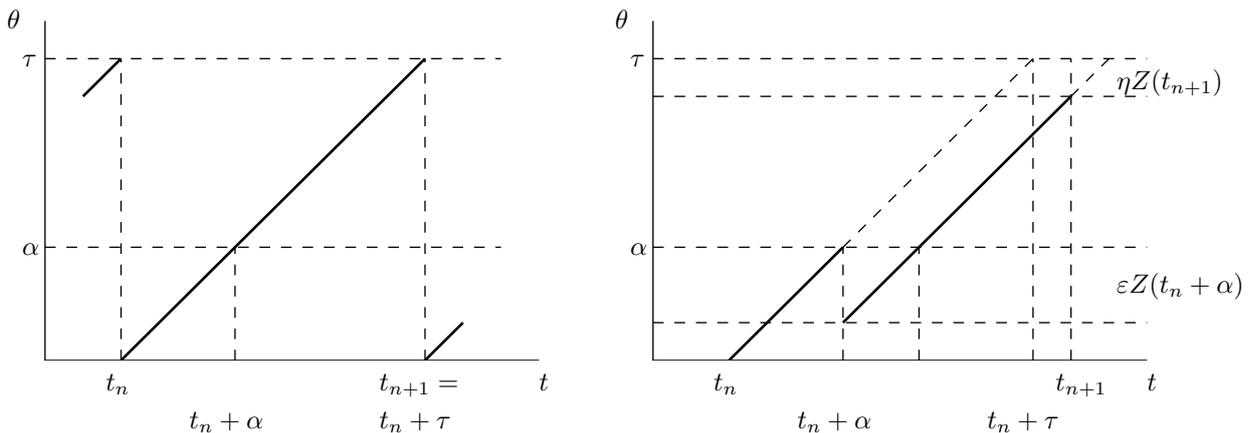}}
\put(13,  1){$t_n$}
\put(24, -4){$t_n + \alpha$}
\put(49,  1){$t_{n+1}=$}
\put(49, -4){$ t_n + \tau$}
\put(70, 1){$t$}
\put(2,  19){$\alpha$}
\put(2,  44){$\tau$}
\put(0,  49){$\theta$}
\put(93,  1){$t_n$}
\put(104,-4){$t_n + \alpha$}
\put(129,-4){$t_n + \tau$}
\put(138, 1){$t_{n+1}$}
\put(150, 1){$t$}
\put(82, 19){$\alpha$}
\put(82, 44){$\tau$}
\put(80, 49){$\theta$}
\put(146,41){$\eta \zgbr(t_{n+1})$}
\put(146,14){$\eps \zgbr(t_n+\alpha)$}
\end{picture}
\caption{\textit{The phase $\theta$ as a function of $t$ without (left) and with (right) Zeitgeber. Suppose the phase $\theta$ is 0 at time $t_n$, then the state of the cell is active for $0 \leq \theta < \alpha$ and inactive for $\alpha \leq \theta < \tau$. Recall that $\alpha+\rho=\tau$. Without Zeitgeber (left figure) this corresponds to the time intervals $[t_n,t_n+\alpha)$ and $[t_n+\alpha,t_n+\tau)$. In presence of a Zeitgeber the cell is active in the time interval $[t_n,t_n+\alpha+\eps \zgbr(t_n+\alpha))$ and inactive in the time interval $[t_n+\alpha+\eps \zgbr(t_n+\alpha),t_{n+1})$, now $t_{n+1}$ is defined implicitly, see text.}\label{fig:theta}}
\end{figure}

We will use the transition times $t_n$ in phase space $\fR$ to define a dynamical system whose evolution takes $t_n$ into $t_{n+1}$. Without Zeitgeber we have a sequence $t_n =t_0 + n \tau$ as described above depending on the parameters $\alpha$ (although trivially) and $\tau$. Next we include a Zeitgeber.

\textbf{A single pacer cell with a periodic Zeitgeber.} To model a non-isolated pacer cell in a collection of other pacer cells with or without an external stimulus, we consider a single pacer cell with a Zeitgeber $\zgbr$ which is a function of time only. First we define the Zeitgeber $\zgbr$. Though non-essential it is convenient to scale time so that the period of the Zeitgeber becomes one.\commentaar{mogelijk hier meer benadrukken dat we periode van Zeitgeber als eenheid van tijd nemen}
\begin{definition}\label{def:extsig}
The positive function $\zgbr : \fR \to [0,1]$ satisfies the following

\begin{enumerate}\iii
\item $\zgbr$ is differentiable,
\item $\zgbr$ is periodic with period $1$.
\end{enumerate}
\end{definition}


In the presence of a Zeitgeber the phase $\theta$ again increases with speed $1$, starting at $\theta=0$ at $t=t_n$, but when $\theta$ reaches the value $\alpha$ it instantly drops back by an amount of $\eps \zgbr(t_n+\alpha)$. Then it again increases with speed $1$ until it reaches a value at $t=t_{n+1}$ such that $\theta(t_{n+1}) + \eta \zgbr(t_{n+1}) = \tau$. Note that $t_{n+1}$ is implicitly defined. Thus the phase $\theta$ of the cell is given by
\begin{equation}\label{eq:theta}
\theta(t) =
\begin{cases}
t-t_n, & \text{for}\; t \in [t_n,t_n+\alpha)\\
t-t_n - \eps \zgbr(t_n+\alpha) & \text{for}\; t \in [t_n+\alpha, t_{n+1})\\
0,                         & \text{for}\; t = t_{n+1}.
\end{cases}
\end{equation}
Apart from the parameters $\alpha$ and $\tau$ we now also have $\eps$ and $\eta$. The latter two give the `strength' of the Zeitgeber. In order that the model be consistent we impose the following conditions on the parameters
\begin{equation*}
0 < \alpha < \tau,\;\; \eps \geq 0,\;\; \eta \geq 0,\;\; \alpha - \eps > 0,\;\; \alpha - \eps + \eta < \tau,
\end{equation*}
so that $\theta$ remains between $0$ and $\tau$.

\textbf{Dynamical system.} The state of the dynamical system we define is the transition time $t_n$ rather than the phase $\theta$. Indeed solving equation
\begin{equation}\label{eq:t}
\theta(t) + \eta \zgbr(t) = \tau \;\;\text{or equivalently}\;\;t-t_n - \eps \zgbr(t_n+\alpha) + \eta \zgbr(t) = \tau
\end{equation}
for $t$, under conditions to be specified later, yields a unique solution $t=t_{n+1}$ once $t_n$ is given. Thus we may write $t_{n+1} = \Fmu(t_n)$ for a map $\Fmu : \fR \to \fR$ depending on parameters $\mu=(\eps,\eta,\alpha,\tau)$. However, it turns out that $\Fmu$ has the property $\Fmu(t+1)=\Fmu(t)+1$ so that $\Fmu$ is the \emph{lift} of a circle map $\fmu : S^1 \to S^1$. This means that we now have a dynamical system with \emph{phase space} the circle $S^1$ and \emph{evolution law} $\fmu$. Conceptually it is easier to work with the circle map $\fmu$ but for actual computations we usually prefer the lift $\Fmu$. Let us summarize the result in the following proposition, for a proof see the appendix.

\begin{proposition}[Circle map and lift]\label{pro:ds}
Let $\zgbr$ be as in definition \ref{def:extsig} and define the function $U_{\eps} :\fR \to \fR$ as $U_{\eps}(t) = t + \eps \zgbr(t)$. Then the map $\Fmu : \fR \to \fR$ with
\begin{equation}\label{eq:dslift}
\Fmu(t) = U_{\eta}^{-1}(U_{\eps}(t+\alpha) - \alpha + \tau)
\end{equation}
defines a parameter dependent differentiable dynamical system, provided that $\eta$ is small enough. Parameters are $\mu = (\eps, \eta, \alpha, \tau)$. Furthermore, $\Fmu$ is the lift of a circle map $\fmu : S^1 \to S^1$ of degree one, given by
\begin{equation}\label{eq:dscmap}
\fmu(t) = U_{\eta}^{-1}(U_{\eps}(t+\alpha) - \alpha + \tau) \mod 1.
\end{equation}
\end{proposition}

A circle map is a one-dimensional map just like an interval map. The dynamics of the two have much in common which is most prominent when the circle map is studied by a lift. However, because the circle is different from the interval\commentaar{global topology of the circle is different from that of the interval}, there are also differences in dynamical behavior. For example non-degenerate fixed points of a circle map come in pairs. See figure \ref{fig:maplift}.

\begin{figure}[htbp]
\setlength{\unitlength}{1mm}
\begin{picture}(100,50)(0,0)
\put(40,  5){\includegraphics{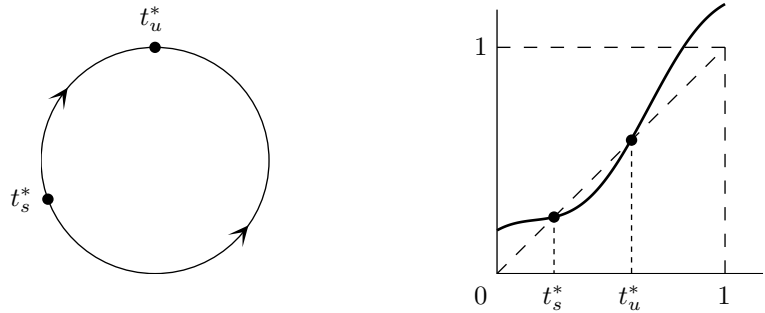}}
\put(97,  1){$0$}
\put(129, 1){$1$}
\put(97, 34){$1$}
\put(106, 1){$t_s^*$}
\put(116, 1){$t_u^*$}
\put(36, 14){$t_s^*$}
\put(53, 38){$t_u^*$}
\end{picture}
\caption{\textit{Phase portrait of circle map $\fmu$ and graph of lift $\Fmu$. $\fmu$ has two fixed points indicated by bullets. One is stable, the other is unstable according to the arrows. The lift $\Fmu$ is a map on the interval $[0,1]$ with $\Fmu(1) = \Fmu(0)+1$, drawn is $\Fmu-1$. $t_s^*$ (stable) and $t_u^*$ (unstable) satisfy $\Fmu(t)=t+1$. The Zeitgeber in this example is $\zgbr(t)=\frac{1}{2}(1+\sin(2\pi t))$.}\label{fig:maplift}}
\end{figure}

We now make a further distinction between two cases, namely whether $\fmu$ is invertible (a diffeomorphism) or not (an endomorphism). The difference between these cases is not only in dynamical behavior but the second case has far richer bifurcation scenarios. Essentially it boils down to both $\eps$ and $\eta$ being 'sufficiently small' or one of them not meeting this criterion. However, there is a priori no reason to assume that either of them is small. \commentaar{See figure \ref{fig:cases} for an illustration.} \opmerking{nog meer over te zeggen?}
\commentaar{treat diffeo and endo on equal footing??}
\begin{enumerate}\abc
\item \label{itm:diffeo}\textbf{$\Fmu$ is the lift of a circle diffeomorphism.} In this case $\Fmu$ is differentiable and $F^{-1}_{\mu}$ exists and is also differentiable. Both $U_{\eta}$ and $U_{\eps}$ have to be invertible, which in turn means that $\eps$ and $\eta$ must be small enough.
\item \label{itm:endo}\textbf{$\Fmu$ is the lift of a circle endomorphism.} In this case $\Fmu$ is again differentiable but $F^{-1}_{\mu}$ does not necessarily exist. Now only $U_{\eta}$ has to be invertible, which means that only $\eta$ must be small enough.
\end{enumerate}
\commentaar{
%
%
}
\begin{remark}
Here $\eps$ 'small enough' means that $U_{\eps}$ is invertible for which we need that $U_{\eps}' > 0$. This depends on the specific form of the Zeitgeber. For the \emph{standard Zeitgeber} $\zgbr(t) = \frac{1}{2}(1+\sin(2 \pi t))$, small enough means $\eps < \frac{1}{\pi}$. Thus the circle map $\fmu$ is a diffeomorphism if both $\eps$ and $\eta$ are smaller than $\frac{1}{\pi}$.
\end{remark}
\begin{remark} If $\eta$ is not small, but $\eps$ is small enough so that $U_{\eps}$ is invertible, $F^{-1}_{\mu}$ is the lift of a circle endomorphism, but $\Fmu$ is multi-valued. This case is similar to case b with the roles of $\Fmu$ and $F^{-1}_{\mu}$ interchanged. However in the dynamical system defined with $F^{-1}_{\mu}$ time is running backwards so it describes the past rather than the future. Mathematically this is not a problem, but the biological interpretation could be problematic. One could force $\Fmu$ to be single-valued by choosing the smallest solution for $t=t_{n+1}$ of equation \eqref{eq:t}. Then $\Fmu$ again defines a dynamical system, though a \emph{discontinuous} one.
\end{remark}
\begin{remark} If both $\eps$ and $\eta$ are not small $\Fmu$ and $F^{-1}_{\mu}$ are multi-valued. In this case we do not have a well-defined dynamical system at all. But a similar construction as in the previous remark can be applied to define a possibly discontinuous dynamical system.
\end{remark}

\section{Analysis of the mathematical model}\label{sec:matmodanl}
Here we restrict ourselves to the case that $\Fmu$ in proposition \ref{pro:ds} is the lift of a circle diffeomorphism $\fmu$, case \ref{itm:diffeo}) on page \pageref{itm:diffeo}. In that case we can use the \emph{rotation number} which tells how much on average an initial point is rotated along the circle by $\fmu$. It is a powerful tool in determining whether $\fmu$ has fixed points or periodic points. The existence of such points and their dependence on the parameters $\mu$ is the main topic of this section. Stable fixed or periodic points are the most relevant for our model and we will see how they lose stability at certain bifurcations. For background on circle maps and further references to the literature see \cite{arn83,bt,dev,kh}.

\subsection{Special cases related to the Arnol'd map}\label{sec:specases}
We begin with an example, namely two special cases where $\Fmu$ can be related to the lift of the Arnol'd or standard circle map $A_{\omega,\lambda}$. Here we make a special choice for the Zeitgeber $\zgbr$, namely $\zgbr(t) = \frac{1}{2}(1+\sin(2\pi t))$. The rotation number of the Arnol'd map, depending on the parameters $\omega$ and $\lambda$ is well studied, so we have quite some information on fixed and periodic points. The (lift of the) Arnol'd map is defined as
\begin{equation*}
A_{\omega,\lambda}(t) = t + \omega + \lambda \sin(2\pi t),
\end{equation*}
for $t \in \fR$ and parameters $\lambda, \omega \in \fR$. In figure \ref{fig:tongharen} we indicate the regions in the $(\omega,\lambda)$-plane where the rotation number of $A_{\omega,\lambda}$ is constant. These regions are called \emph{tongues} and the general theory of circle diffeomorphisms tells us that for parameter values inside the tongues the map has stable fixed or periodic points. On the tongue boundaries we have saddle-node bifurcations of fixed points in the main tongue or periodic points in the other tongues, see figure \ref{fig:tongharen}.

Recall that the lift of our circle map $\fmu$ is given by $\Fmu(t) = U^{-1}_{\eta}(U_{\eps}(t+\alpha)-\alpha+\tau)$, where $U_{\eps}(t) = t + \eps \zgbr(t)$ and $\zgbr$ is the periodic forcing. For a special choice of $\zgbr$ and parameters $\mu$, $\Fmu$ transforms into the Arnol'd map by a change of coordinates.\commentaar{also see remark \ref{rem:finv} on page \pageref{rem:finv}} We summarize this in the following lemma.
\begin{lemma}[Conjugation to Arnol'd map]\label{lem:ftoa}
Let the periodic Zeitgeber $\zgbr$ be given by $\zgbr(t) = \frac{1}{2}(1+\sin(2\pi t))$, then
\begin{enumerate}\itemsep 0pt
\item the map $F_{(\eps,0,\alpha,\tau)}$ is conjugate to the Arnol'd map $A_{\tau+\frac{1}{2} \eps, \frac{1}{2} \eps}$, where the conjugation is a rigid translation over $\alpha$,
\item the map $F^{-1}_{(0,\eta,\alpha,\tau)}$ is equal to the Arnol'd map $A_{-\tau+\frac{1}{2} \eta, \frac{1}{2} \eta}$.
\end{enumerate}
\end{lemma}
Using $F^{-1}$ in the second part of the lemma may seem unnatural in the present context, but if $\rho$ is the rotation number of $F^{-1}$, then $1-\rho$ is the rotation number of $F$. Therefore the second part yields information about the rotation number of $F_{(0,\eta,\alpha,\tau)}$.\commentaar{Also compare with remark \ref{rem:finv} at the end of section \ref{sec:matmod}} Note that the parameter transformation from parameters of $F$ to those of $A$ does not involve $\alpha$. Thus $\alpha$ does not play a role in the bifurcation analysis of these two special cases. The proof of the lemma is straightforward and therefore omitted.
\begin{figure}[htbp]
\setlength{\unitlength}{1mm}
\begin{picture}(100,45)(0,0)
\put(5,   5){\includegraphics{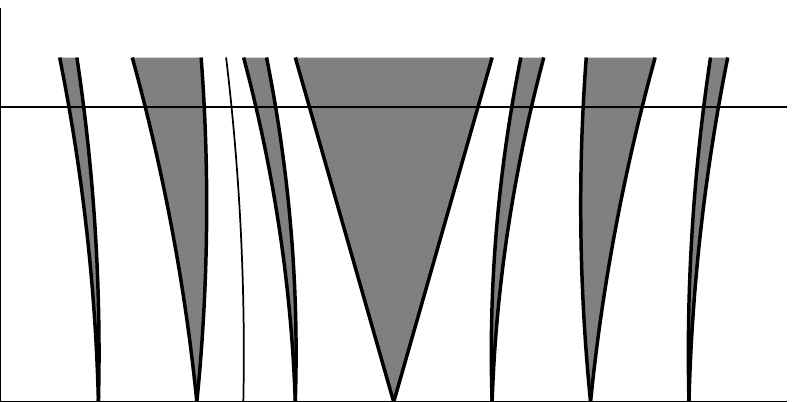}}
\put(95,  5){\includegraphics{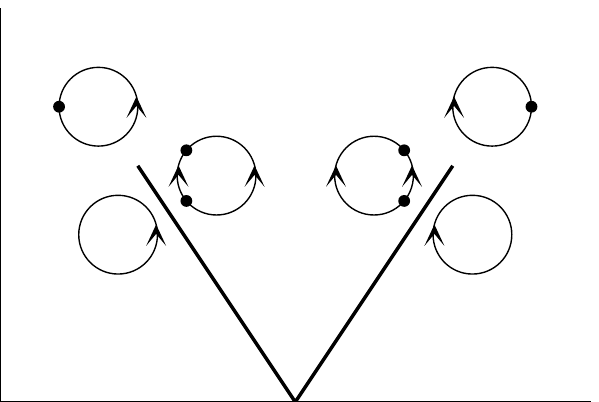}}
\put(80,  2){$\omega$}
\put(0,  42){$\lambda$}
\put(0,  34){$\lambda_0$}
\put(150, 2){$\omega$}
\put(91, 42){$\lambda$}
\put(44,  1){1}
\put(124, 1){1}
\put(12,  1){$p/q$}
\put(28,  1){$\omega_0$}
\end{picture}
\caption{\textit{Left: schematic picture of tongues and hairs. Main tongue emanating at 1, a $p:q$ tongue at $p/q$ and a hair at $\omega_0$, on the $\omega$-axis for the Arnol'd map $A_{\omega,\lambda}$. Right: schematic phase portraits of the same family near and on the boundary of the main tongue. The rotation number as a function of $\omega$ for a fixed value $\lambda=\lambda_0$ is shown in figure \ref{fig:dutrap}.}\label{fig:tongharen}}
\end{figure}
\commentaar{meer haren tekenen}

To analyze these two special cases it suffices to consider the Arnol'd map $A_{\omega,\lambda}$. Let us summarize the properties of the latter. If $\lambda=0$ the map reduces to a rigid rotation $A_{\omega,0} = R_{\omega}$ and therefore the rotation number is $\rho(A_{\omega,0}) = \rho(R_{\omega}) = \omega$. This is a degenerate situation. But for $\lambda \neq 0$ the map is no longer degenerate. Let us first fix $\omega=\omega_0$ and $\lambda=\lambda_0 \neq 0$ such that $A_{\omega_0,\lambda_0}$ has a rational rotation number $\frac{p}{q}$. Then $A_{\omega_0,\lambda_0}$ has $q$-periodic points. If these points are hyperbolic then there is an open neighborhood of $(\omega_0,\lambda_0)$ in the parameter plane such that for all $(\omega,\lambda)$ in this neighborhood, $A_{\omega,\lambda}$ has rotation number $\frac{p}{q}$.

Let us now consider the line $L_{\lambda_0}=\{(\omega,\lambda_0) \;|\; \lambda_0 \neq 0, \omega \in [\frac{1}{2},\frac{3}{2}) \}$ in the parameter plane of the Arnol'd circle map. From the arguments above it follows that this line segment contains open intervals on which the rotation number of the map equals $\frac{p}{q}$. These intervals have to shrink to points when $\lambda_0$ tends to zero, because the rotation number of $A_{\omega,0}$ equals $\omega$. When we consider the rotation number on the line segment $L_{\lambda_0}$ as a function of $\omega$, its graph, see figure \ref{fig:dutrap}, is a so called \emph{devil's staircase}, see \cite{dev, kh} for a definition. A further analysis shows that there are saddle-node bifurcations of $q$-periodic points on the boundary of the intervals in $L_{\lambda_0}$, see \cite{kuz}. Since the map $A_{\omega,\lambda}$ depends differentiable on the parameters there are differentiable pairs of saddle-node curves in the $(\omega,\lambda)$-parameter plane emanating from rational points on the line segment $L_0$. This forms the structure of \emph{tongues} in the parameter plane, see figure \ref{fig:tongharen}. The tongues where the rotation number is $\frac{p}{q}$ are called \emph{$p:q$-tongues} and the tongue where the rotation number is 1 (or 0) is called the \emph{main tongue}.
\begin{figure}[htbp]
\setlength{\unitlength}{1mm}
\begin{picture}(100,40)(0,0)
\put(50,5){\includegraphics{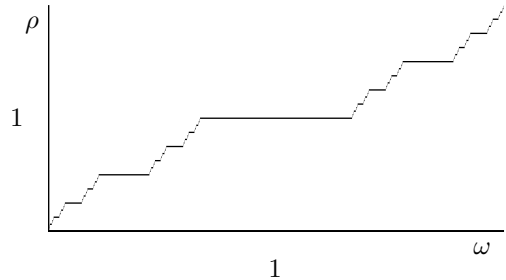}}
\put(106,2){$\omega$}
\put(47,32){$\rho$}
\put(45,19){$1$}
\put(79,-1){$1$}
\end{picture}
\caption{\textit{Devil's staircase: the graph of the rotation number $\rho$ as a function of $\omega$ for the Arnol'd map at a fixed value $\lambda=\lambda_0$. In the main tongue the rotation number is equal to $1$. Also see figure \ref{fig:tongharen}.}\label{fig:dutrap}}
\end{figure}

A rational rotation number $\frac{p}{q}$ for the map $A_{\omega,\lambda}$ is constant on closed intervals on the line segment $L_{\lambda_0}$. But the complement of the union of these closed intervals in $L_{\lambda_0}$ is not empty, it contains points where the rotation number is irrational. Again when $\lambda_0$ tends to zero there are smooth curves with sufficiently irrational rotation number ending in irrational points of $L_0$. These curves are sometimes called \emph{hairs}. The $(\omega,\lambda)$-parameter plane of the Arnol'd circle map consists mainly of \emph{tongues} and \emph{hairs}, see figure \ref{fig:tongharen}. The tongues and hairs fill a relatively large region in the parameter plane, when we take a point $(\omega,\lambda)$ at random, there is a positive probability that it belongs to a tongue, but there is a positive probability as well that it lies on a hair.

\subsection{A standard form for the circle diffeomorphism $\fmu$}\label{sec:standard}
The purpose of this section is to show that there is a standard form for every circle diffeomorphism. The Arnol'd map for example already is in this form. Identifying the standard form for the map $\fmu$ allows us to conclude that it has a tongues-and-hairs structure in a certain parameter plane for every 1-periodic Zeitgeber. Thus the standard form mainly serves a theoretical purpose, for actual computations it is far more advantageous to use the expression in equation \eqref{eq:dslift} for the lift of $\fmu$.

Theorem \ref{the:standard} below shows that every lift $C$ of a differentiable circle map of degree one \commentaar{voor de volledigheid...} can be written in the following form
\begin{equation*}
C(t) = t + P(t),
\end{equation*}
where $P$ is a 1-periodic function. We assume that $P$ is non-zero and non-constant. The Fourier coefficients of $P$ may be considered as parameters. Let $\avg{P}$ denote the average of $P$. Set $\omega = \avg{P}$ and $P_0 = P - \avg{P}$. Then $\omega$ is the constant term of the Fourier series of $P$ and $P_0$ has zero average. Consequently $\lambda = \max_{[0,1]}|P_0|$ is non-zero and we may set $P_{01}=\frac{1}{\lambda}P_0$. Finally we may write
\begin{equation}\label{eq:gencmap}
C_{\omega,\lambda}(t) = t + \omega + \lambda P_{01}(t).
\end{equation}
If we now interpret $\omega$ and $\lambda$ as parameters, then there is a tongues-and-hairs structure in the $(\omega,\lambda)$-plane. Note that we recover the Arnol'd map by setting $P_{01}(t) = \sin(2\pi t)$. In the more general family of equation \eqref{eq:gencmap} the tongues may have a richer structure than those of the Arnol'd map. For example there may be more saddle-node curves inside a tongue, as we will see in the next section. For a detailed description of such phenomena see \cite{bst}.\commentaar{perturbation of rigid rotation, possibly non-small}

Recall that $\Fmu(t) = U^{-1}_{\eta}(U_{\eps}(t+\alpha)-\alpha+\tau)$, where $U_{\eps}(t) = t + \eps \zgbr(t)$ is a lift of the circle diffeomorphism $\fmu$. First we give a standard form for the map $\Fmu$.

\begin{theorem}[Standard form]\label{the:standard}
Suppose that $\eta$ is small enough so that $U^{-1}_{\eta}$ exists. Let $\nu = (\sigma,\beta,\alpha,\tau)$ be new coordinates in parameter space with $\eps = \sigma \cos \beta$ and $\eta = \sigma \sin \beta$. Then there are smooth functions $\omega$ and $\lambda$ of the parameters $\nu$ with $\omega(\nu)|_{\sigma=0} = \tau$, $\lambda(\nu) = \sigma$ and a 1-periodic smooth function $R_{\nu}$ with zero average, smoothly depending on parameters $\nu$ such that 
\begin{equation*}
F_{\mu(\nu)}(t) = t + \omega(\nu) + \lambda(\nu) R_{\nu}(t).
\end{equation*}
\end{theorem}
The theorem shows that after a transformation of parameters, $\Fmu$ takes the standard form of a circle map. In particular we have the result that there is a tongues-and-hairs structure in the $(\tau, \sigma)$-plane, for each value of $\beta \in [0,\frac{\pi}{2})$ and each value of $\alpha \in [0,\tau)$.

The function $R$ depends on the Zeitgeber $\zgbr$, but also on all parameters $\nu$. The latter will appear in the coefficients of the Fourier series of $R$.

\begin{remark}
Setting $\eta=0$ in $\mu=(\eps,\eta,\alpha,\tau)$ corresponds to setting $\beta=0$ in $\nu = (\sigma,\beta,\alpha,\tau)$. Let $\avg{\zgbr} = \int_0^1 \zgbr(t)\,dt$ be the average of $\zgbr$ then we have
\begin{equation*}
F_{(\eps,0,\alpha,\tau)} = t + \tau + \eps \zgbr(t+\alpha) = t + \tau + \eps \avg{\zgbr} + \eps (\zgbr(t+\alpha)-\avg{\zgbr}).
\end{equation*}
If we also take $\zgbr(t) = \frac{1}{2}(1+\sin(2\pi t)$ we recover the result of lemma \ref{lem:ftoa}.
\end{remark}

\begin{remark}
We expect that for most values of $\beta$ the tongues in the $(\tau,\sigma)$-parameter plane are non-degenerate, that is they intersect transversely at the vertices. We even expect this for $\zgbr(t) = \frac{1}{2}(1+\sin(2\pi t))$, since in the standard form of $F$, $R_{\nu}$ will have a Fourier series rather than a Fourier polynomial. Also see \cite{bst}.
\end{remark}

Let us now look at the position of the main tongue. We assume that the Zeitgeber has the following Fourier series
\begin{equation}\label{eq:zgbr}
\zgbr(t) = c_0 + c_1 \sin(2\pi t) + \sum_{k>1} c_k \sin(2\pi(k t + \gamma_k)),
\end{equation}
with coefficients $c_k \in \fR$ and $\gamma_k \in [0,1]$. Note that by a time shift we can always achieve that $\gamma_1 = 0$. Then we have the next result.
\begin{proposition}[Boundaries of main tongue]\label{pro:bmt}
Let the Zeitgeber be as in equation \ref{eq:zgbr}, then the boundaries of the main tongue of the map $F_{\mu(\nu)}$ of are given by
\begin{equation}\label{eq:mt}
\tau_{\pm} = 1 - \sigma \Big[c_0 (\cos\,\beta - \sin\,\beta) \mp c_1 \sqrt{1-\cos(2\alpha\pi) \sin(2\beta)}\Big] + \cO(\sigma^2).
\end{equation}
\end{proposition}
For our standard example of a Zeitgeber $\zgbr(t) = \frac{1}{2}(1+\sin(2\pi t))$ the boundaries of the main tongue are given by equation \eqref{eq:mt} with $c_0=c_1=\frac{1}{2}$ and without the $\cO(\sigma^2)$ term because all other coefficients are equal to zero.

\subsection{Fixed points of the diffeomorphism $\fmu$}\label{sec:fixpo}
Existence of stable fixed points of the circle diffeomorphism $\fmu$ is one of the main questions. Recall that such points correspond to entrainment in the biological model. Therefore we take a closer look at such points. From the previous sections we know that fixed points exist for parameter values in the main tongue. They are easily characterized as follows, $t=t^*$ is a fixed point of $\fmu$ if $\fmu(t^*) = t^*$. However, as noted before, for practical computations it is more convenient to use a lift $\Fmu$ of $\fmu$. Using the notation from proposition \ref{pro:ds} a point $t=t^*$ is a fixed point of $\fmu$ if
\begin{equation*}
\Fmu(t^*) = t^* + 1,
\end{equation*}
 which is equivalent to
\begin{equation}\label{eq:tster}
U_{\eps}(t^*+\alpha)-\alpha+\tau = U_{\eta}(t^*+1).
\end{equation}
Even if $\eps$ and $\eta$ are small enough so that $U_{\eps}$ and $U_{\eta}$ are invertible, this equation has several solution branches. Therefore we can not solve \eqref{eq:tster} explicitly for $t^*$. Since we are mostly interested in the dependence of $t^*$ on $\tau$ we shall content ourselves with
\begin{equation}\label{eq:tau}
\tau = U_{\eta}(t^*) - U_{\eps}(t^*+\alpha) + \alpha + 1.
\end{equation}
Thus for fixed values of $\eps$, $\eta$ and $\alpha$ we obtain $\tau$ as a function of $t^*$.

Now that we have characterized the fixed points of $\fmu$, let us determine their stability. The fixed point $t=t^*$ is stable if $|F'_{\mu}(t^*)| < 1$. After a short computation we find
\begin{equation}\label{eq:df}
F'_{\mu}(t^*) = \frac{U'_{\eps}(t^*+\alpha)}{U'_{\eta}(t^*)}.
\end{equation}
From equations \eqref{eq:tau} and \eqref{eq:df} we almost immediately obtain an alternative characterization of stability of the fixed point $t^*$. Two examples of application of the following lemma are given in figure \ref{fig:snbifs}.
\begin{lemma}[Stability]\label{lem:stab}
On solution branches of equation \eqref{eq:tster}, the fixed point $t^*$ is stable or unstable when $t^*$ as a function of $\tau$ is increasing or decreasing.
\end{lemma}

Let us again consider our standard Zeitgeber $\zgbr(t) = \frac{1}{2}(1+\sin(2\pi t))$. For fixed $\eps$, $\eta$ and $\alpha$ there are two values $t_1^*$ and $t_2^*$ for which $F'_{\mu}(t_i^*)| = 1$ and $|F'_{\mu}(t^*)| < 1$ if $t^* \in [t_1^*, t_2^*]$.\commentaar{For this particular Zeitgeber we even have $t_2^* = t_1^* + \frac{1}{2}$.} We recover the boundaries of the main tongue by setting $\tau_i = U_{\eta}(t_i^*) - U_{\eps}(t_i^*+\alpha) + \alpha + 1$, then a stable and an unstable fixed point of the map $\fmu$ exist for $\tau$ on the interval $[\tau_1,\tau_2]$. Motivated by this example one could conjecture that for every 1-periodic Zeitgeber with only two extrema in one period there are precisely two fixed points for parameter values in the main tongue. This turns out to be true when an extra condition is imposed.
\begin{proposition}[Number of fixed points]\label{pro:fixpos}
Let $\zgbr$ be a 1-periodic Zeitgeber with one maximum and one minimum on the interval $(0,1)$. Then the circle map $\fmu$ has precisely two fixed points for parameter values in the main tongue if $(\zgbr'')^2 - \zgbr'\cdot\zgbr'''$ does not change sign.
\end{proposition}
For our standard example of a Zeitgeber $\zgbr(t) = \frac{1}{2}(1+\sin(2\pi t))$ the quantity $(\zgbr'')^2 - \zgbr'\cdot\zgbr'''$ is equal to 1. The Zeitgeber $\zgbr(t) = \frac{2}{5}(\frac{4}{3} + \sin(2\pi t) + \frac{1}{3}\sin(4\pi t))$ also has two extrema, but the quantity $(\zgbr'')^2 - \zgbr'\cdot\zgbr'''$ changes sign on $[0,1]$. See the appendix for further implications. With a general periodic Zeitgeber there may be more than one stable fixed point for parameter values in the main tongue. This occurs in general when the Fourier series of the Zeitgeber contains more than just one term like in our standard example. Let us look at a Zeitgeber with the following Fourier polynomial $\zgbr(t) = \frac{3}{10}(2+\sin(2\pi t)+\cos(4\pi t))$, then there are four solutions of $F'_{\mu}(t_i^*)| = 1$. Let us sort them such that $\tau_i=U_{\eta}(t_i^*) - U_{\eps}(t_i^*+\alpha) + \alpha + 1$ increases with $i$. Then for fixed values of $\eps$, $\eta$ and $\alpha$ there are at most four fixed points of $\fmu$ for each $\tau \in [\tau_1, \tau_4]$. We graphically represent the results in figure \ref{fig:snbifs}.
\begin{figure}[htbp]
\setlength{\unitlength}{1mm}
\begin{picture}(100,50)(0,0)
\put(10,  5){\includegraphics{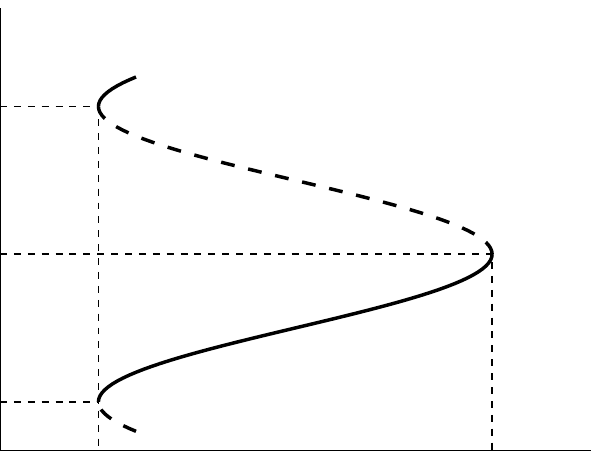}}
\put(5,  47){$t^*$}
\put(6,   9){$t_1^*$}
\put(6,  24){$t_2^*$}
\put(-1, 39){$t_1^*+1$}
\put(67,  0){$\tau$}
\put(19,  2){$\tau_1$}
\put(59,  2){$\tau_2$}
\put(90,  5){\includegraphics{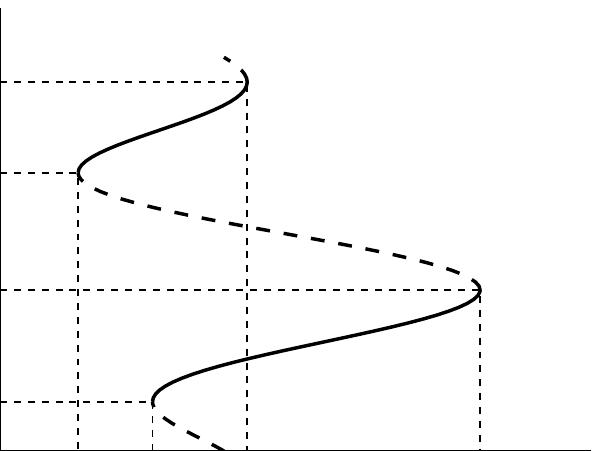}}
\put(85, 47){$t^*$}
\put(86 ,32){$t_1^*$}
\put(86,  9){$t_2^*$}
\put(86, 41){$t_3^*$}
\put(86, 20){$t_4^*$}
\put(147, 0){$\tau$}
\put(97,  2){$\tau_1$}
\put(104, 2){$\tau_2$}
\put(114, 2){$\tau_3$}
\put(138, 2){$\tau_4$}
\end{picture}
\caption{\textit{Fixed points $t^*$ of the map $\fmu$ as a function of $\tau$. The solid curve represents a stable fixed point, the dashed curve an unstable one. At parameter values $\tau=\tau_i$ there are saddle-node bifurcations, also see figure \ref{fig:tongharen}. Parameters $\eps$, $\eta$ and $\alpha$ are fixed. Left: Zeitgeber is $\zgbr(t) = \frac{1}{2}(1+\sin(2\pi t))$. Right: Zeitgeber is $\zgbr(t) = \frac{3}{10}(2+\sin(2\pi t)+\cos(4\pi t))$.}\label{fig:snbifs}}
\end{figure}

\subsection{Examples of tongues-and-hairs for different Zeitgebers}\label{sec:thz}
Here we collect some examples of tongues-and-hairs figures showing differences and similarities with the prototype figure of tongues in the Arnol'd circle map. Therefore we start with the latter, see figure \ref{fig:arnoldtongues}. In this family the tongues are fixed, but in our family $\fmu$, the tongues in the $(\tau,\sigma)$-plane still depend on the values of $\beta$ and $\alpha$. Moreover they also depend on the Zeitgeber. In order to keep the number of pictures limited we only show tongues for two different Zeitgebers, see figures \ref{fig:standardtongues} and \ref{fig:specialtongues}. As is to be expected form the existence of a standard form, see section \ref{sec:standard}, the pictures are qualitatively the same. That is near the $\tau$-axis or $\omega$-axis in the Arnol'd map. The width and the growth of the width of the tongues depends on parameters in the map. This is shown in figure \ref{fig:range}, in the next section, where its biological relevance is discussed.

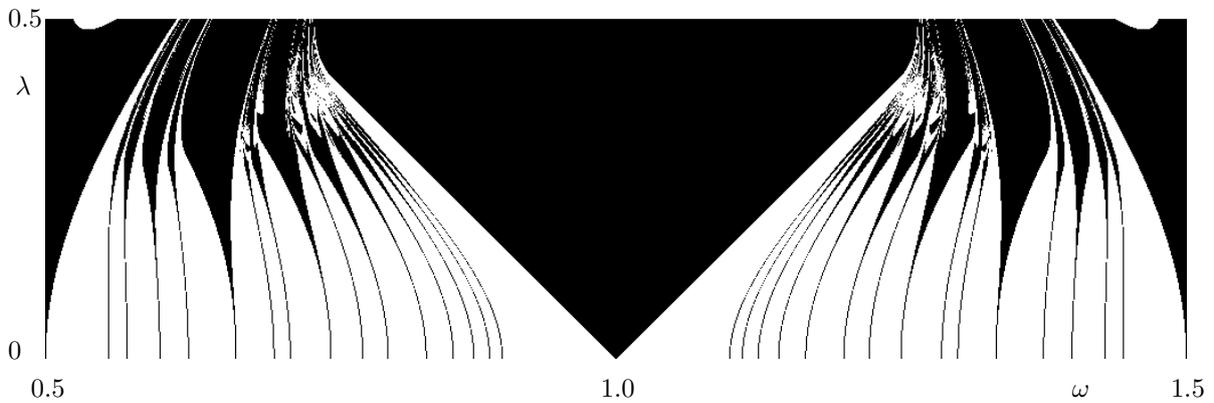
\begin{figure}[htbp]
\setlength{\unitlength}{1mm}
\begin{picture}(100,50)(0,0)
\put(5,5){\includegraphics[width=15cm]{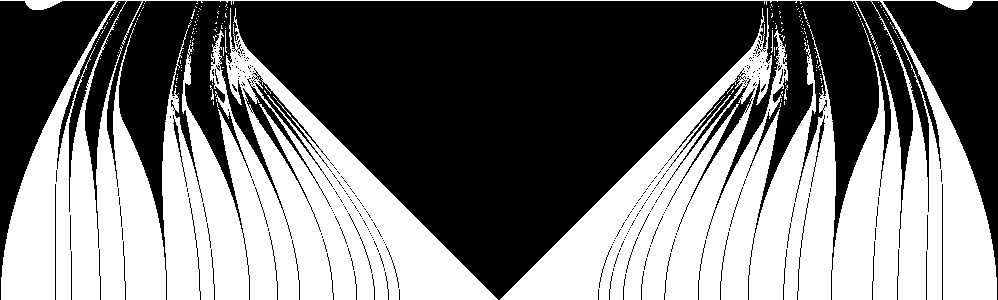}}
\put(140,0){$\omega$}
\put(1, 40){$\lambda$}
\put(0,  5){$0$}
\put(0, 49){$0.5$}
\put(3,  0){$0.5$}
\put(78, 0){$1.0$}
\put(153,0){$1.5$}
\end{picture}
\caption{\textit{Tongues for the Arnol'd map $A_{\omega,\lambda}$ in the $(\omega,\lambda)$-plane. The algorithm to compute the pictures is based on the rotation number. The latter is defined for sufficiently small values of $\lambda$ only. Therefore the tongue boundaries in the picture are not well defined for large values of $\lambda$. The effect of this phenomenon is even more visible in the following picture.}\label{fig:arnoldtongues}}
\end{figure}

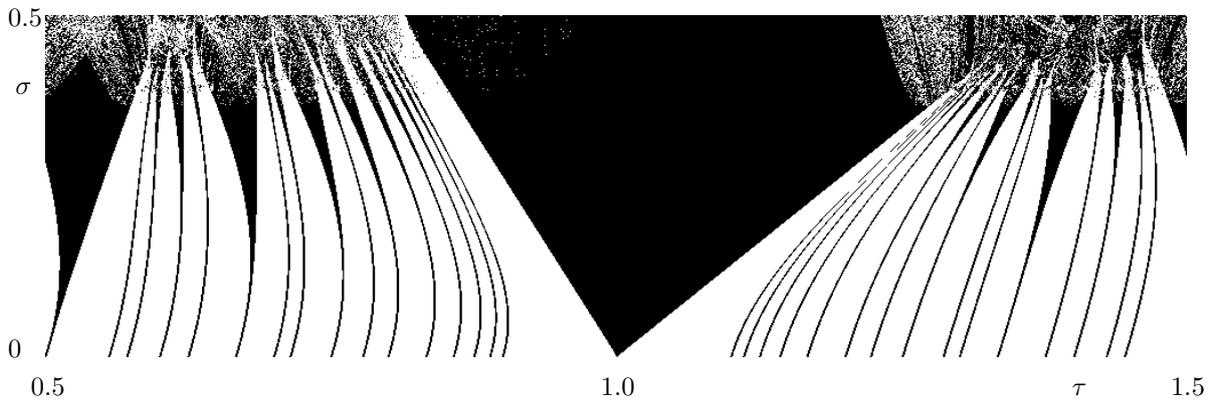
\begin{figure}[htbp]
\setlength{\unitlength}{1mm}
\begin{picture}(100,50)(0,0)
\put(5,5){\includegraphics[width=15cm]{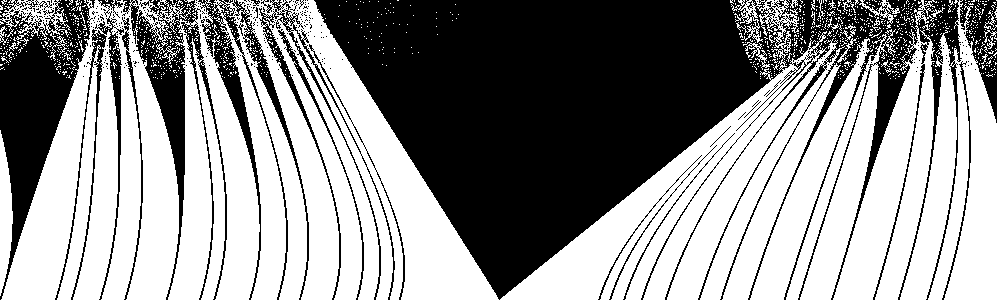}}
\put(140,0){$\tau$}
\put(1, 40){$\sigma$}
\put(0,  5){$0$}
\put(0, 49){$0.5$}
\put(3,  0){$0.5$}
\put(78, 0){$1.0$}
\put(153,0){$1.5$}
\end{picture}
\caption{\textit{Tongues for the circle map $\fmu$ in the $(\tau,\sigma)$-plane, with standard Zeitgeber $\zgbr(t) = \frac{1}{2}(1+\sin(2\pi t))$. The values of the parameters $\mu=(\eps,\eta,\alpha,\tau)$ are $\eps=\sigma \cos\,\beta$, $\eta=\sigma \sin\,\beta$ with $\beta=\frac{\pi}{3}$ and $\alpha=0.3$.}\label{fig:standardtongues}}
\end{figure}

\begin{figure}[htbp]
\setlength{\unitlength}{1mm}
\begin{picture}(100,50)(0,0)
\put(5,5){\includegraphics[width=15cm]{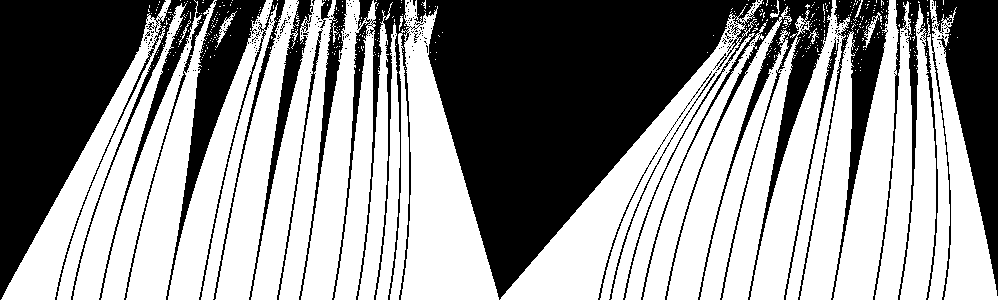}}
\put(140,0){$\tau$}
\put(1, 40){$\sigma$}
\put(0,  5){$0$}
\put(0, 49){$0.5$}
\put(3,  0){$0.5$}
\put(78, 0){$1.0$}
\put(153,0){$1.5$}
\end{picture}
\caption{\textit{Tongues for the circle map $\fmu$ in the $(\tau,\sigma)$-plane, with Zeitgeber $\zgbr(t) = \frac{3}{10}(2+\sin(2\pi t)+\cos(4\pi t))$. Parameter values as in figure \ref{fig:standardtongues}.}\label{fig:specialtongues}}
\end{figure}
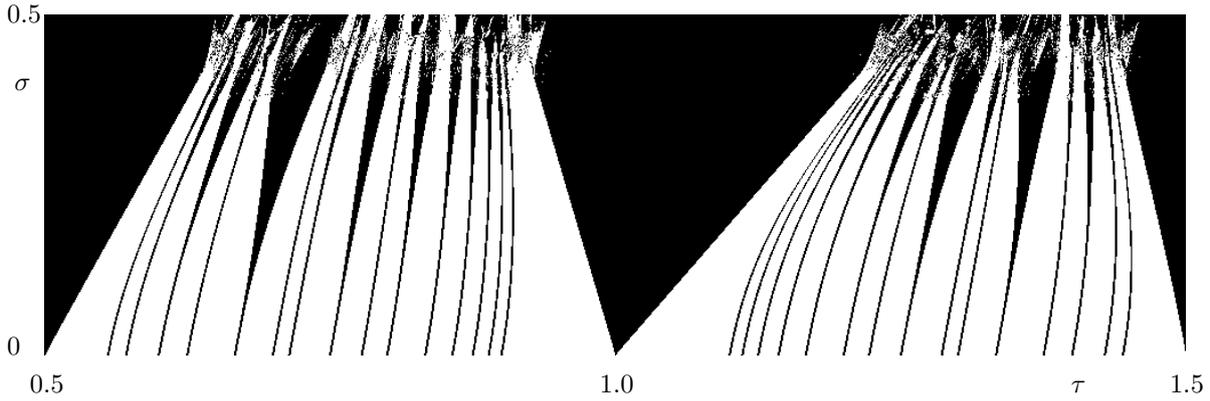

\section{Discussion and conclusion}\label{sec:dc}

\commentaar{Biological interpretation of results from mathematical model, what is explained, what is predicted, comparison with biol data, do we need new math model for same biol model?, do we need new biol model? Or to early to say? Outlook: further investigation within scope of present math model, further investigation outside scope of present math model}

\commentaar{
\begin{itemize}\itemsep 0pt
\item[-] Model in globale termen, voor luie lezers
  \begin{itemize}\itemsep 0pt
    \item[-] cel met eigen periode en periodieke Zeitgeber, cirkel afbeelding
    \item[-] fase ruimte en parameter ruimte (parameters in de Zeitgeber?)
    \item[-] biologische betekenis der parameters (tau en T, experimentele toegankelijkheid)
    \item[-] diffeomorfismen nu en endomorfismen later
  \end{itemize}
\item[-] Dynamica in globale termen, voor luie lezers
  \begin{itemize}\itemsep 0pt
    \item[-] vaste en periodieke punten, quasi periodieke dyn
    \item[-] tongen in parameter ruimte
    \item[-] tongranden en SN-bifs, PF-bifs in de tongen
    \item[-] dynamica buiten de tongen
    \end{itemize}
\item[-] Puntsgewijze bespreking met biologische duiding
  \begin{itemize}\itemsep 0pt
    \item[-] vaste punten, hoofdtong, breedte en scheefheid van deze ('range of entrainment'), poolcoordinaten
    \item[-] periodieke punten
    \item[-] SN-bifs en/of PF-bifs in hoofdtong (experimenten)
    \item[-] quasi periodieke dynamica
  \end{itemize}
\item[-] Biologische observaties van zoogdieren, al dan niet met huidige model te verklaren
  \begin{itemize}\itemsep 0pt
    \item[-] in lengte varierende lichtpulsen en muis wiens ritme volgt tot deze ineens op ander ritme overgaat, ds verklaring mogelijk met meer stabiele takken en PF-bif als organiserend centrum, hiervoor is een Zeitgeber met parameter nodig
    \item[-] daglengte variatie (op polaire schaal) en muis die zijn activiteit interval aanzienlijk verlegt, ds verklaring is bv een zadel knoop bif
    \item[-] 'relative coordination', act-inact cyclus volgt schijnbaar een periode, na een tijd een andere, ds verklaring kan zijn dat er van zadel naar zadel gegaan wordt, of aanloopverschijnsel
    \item[-] ...
  \end{itemize}
\item[-] Toekomstige richtingen
  \begin{itemize}\itemsep 0pt
    \item[-] Huidig model (periodieke Zeitgeber)
    \begin{itemize}\itemsep 0pt
      \item[-] diffeomorfisme
      \item[-] endomorfisme
    \end{itemize}
    \item[-] Huidig model, quasi periodieke Zeitgeber
    \item[-] Model voor meer pacer cellen
    \begin{itemize}\itemsep 0pt
      \item[-] geen interactie maar periodieke Zeitgeber (kan met huidig model)
      \item[-] alleen onderlinge interactie
      \item[-] interactie en periodieke Zeitgeber
      \item[-] interactie en quasi periodieke Zeitgeber
    \end{itemize}
  \end{itemize}
\end{itemize}
}

\subsection*{Model}\label{sec:dcmod}
The biological model for a single pacer cell with a periodic Zeitgeber as described in section \ref{sec:biomod} leads to quantitative mathematical model, namely a dynamical system, which we now summarize. The heart of the mathematical model is a circle map $\fmu$ that yields once $t_0$ is given, a sequence of relative times $t_n$ corresponding to the beginning of an activity interval of the pacer cell. The sequence $t_0, t_1=\fmu(t_0),t_2=\fmu(t_1),\ldots$ is called the itinerary of $t_0$. In this setting time is relative to the period $T$ of the periodic Zeitgeber that acts as a stimulus to the pacer cell. It is convenient to use $T$ as a unit of time, that is we scale time so that $T=1$. The circle map $\fmu$ depends on parameters $\mu = (\eps, \eta, \alpha, \tau)$, where $\tau$ is the intrinsic period of the pacer cell, measured in time unit $T=1$. Furthermore $\eps$ and $\eta$ determine phase delay and phase advance and $\alpha$ determines the length of the intrinsic activity interval, see section \ref{sec:biomod} for an explanation of these terms.

\commentaar{Note on comparing results from experiments and observations on animals with results from the model for a single pacer cell. Model for collection of pacer cells will have richer dynamical behavior than current model.}

\subsection*{Dynamics}\label{sec:dcdyn}
In the description of the dynamical behavior of the map $\fmu$ there is a difference between $\fmu$ being invertible (diffeomorphic) or not. We mainly restrict to the invertible case. Furthermore we restrict to \emph{typical dynamics}, see \cite{bt} for a precise definition. In the present case this comprises fixed points, periodic points and quasi-periodic points. Note that the itinerary of the latter consists of an infinite sequence.

The simplest kind of dynamics is a fixed point of $\fmu$ which means that the activity interval of the pacer cell always starts at the same relative time. The existence of such points implies the possibility of entrainment. The next simplest kind of dynamics is a periodic point of $\fmu$. This means that there is a $t_0$ such that in the sequence $t_0, t_1=\fmu(t_0),\ldots,t_q=\fmu(t_{q-1})$ the last point $t_q$ is again equal to $t_0$ for some fixed $q$. In general these $q$ onsets of activity occur in $p$ periods of the Zeitgeber. We call this synchronization of the pacer cell, a generalization of entrainment. The periodic point $t_0$ is called $p:q$ periodic. The last kind of typical dynamics for the map $\fmu$, is quasi periodicity. The point $t_0$ is quasi-periodic if the itinerary of $t_0$ densely fills the circle or interval $[0,1]$ depending on how we represent the map.

The kind of dynamics the map $\fmu$ exhibits depends on the values of the parameters. All this is nicely organized in parameter space in wedge shaped regions called tongues, one for each pair $(p,q)$, see figure \ref{fig:tongharen}. However, here we need the restriction that $\fmu$ is invertible. For parameter values in the $(p,q)$ tongue, the typical dynamics of $\fmu$ is $p:q$-periodicity. A special role is played by the tongue for $(p,q)=(1,1)$ called the main tongue. Fixed points are the typical dynamics of $\fmu$ for parameter values in the main tongue. At the boundaries of the tongues the map $\fmu$ has a saddle-node bifurcation, see section \ref{sec:specases}. For parameter values outside the tongues the dynamics of the map is quasi-periodic, this occurs on hairs in the $(\tau,\sigma)$-parameter plane. If $\fmu$ is not invertible, tongues still exist but generally overlap so that coexistence of periodic points with different periods becomes possible.

\subsection*{Entrainment}\label{sec:dcent}
For parameter values in the main tongue, the tongue with $(p,q)=(1,1)$, the map $\fmu$ has fixed points. There are at least two such points of which one is stable and the other is unstable. The stable one corresponds to entrainment of the pacer cell. Let us discuss the shape of the main tongue in some more detail to find the parameter values for which entrainment occurs.

The parameter space is four dimensional and in it tongue boundaries are hyper-surfaces. The situation becomes simpler if we do not use coordinates $(\eps,\eta,\alpha,\tau)$ but $(\sigma,\beta,\alpha,\tau)$ with $\eps=\sigma\cos\beta$ and $\eta=\sigma\sin\beta$. Then $\sigma$ satisfies $\sigma^2 = \eps^2 + \eta^2$ and measures how strongly the Zeitgeber stimulates the pacer cell, while $\beta$ determines the ratio of $\eps$ and $\eta$. In these coordinates in parameter space the boundaries of the main tongue are given by equation \eqref{eq:mt}. As we can see from the standard form in theorem \ref{the:standard} the main parameters are $\tau$ and $\sigma$. In the $(\tau,\sigma)$-plane we find the tongues, see figure \ref{fig:range}, whose detailed shape depends on $\alpha$ and $\beta$. The \emph{range of entrainment} is given by the interval $(\tau_-,\tau_+)$, with $\tau_{\pm}$ given in equation \eqref{eq:mt}, when other parameters are kept fixed. There are many biological experiments/observations supporting the existence of bounded ranges of entrainment, see \cite{da01}. In practice one cannot vary the intrinsic period $\tau$ of the pacer cell. But varying the period of the Zeitgeber $T$ at a fixed value of $\tau$ is equivalent to varying $\tau$ and fixing $T$ in the model, see \cite{ap78,rm01,uto}. However, biological evidence exists that the intrinsic period varies among individual pacer cells while they can still entrain to a Zeitgeber with a 24-hour period, see \cite{hhh,hnsh,hsknh,shkoh,wlmr}.

\begin{figure}[htbp]
\setlength{\unitlength}{1mm}
\begin{picture}(100,45)(0,0)
\put(20, 5){\includegraphics{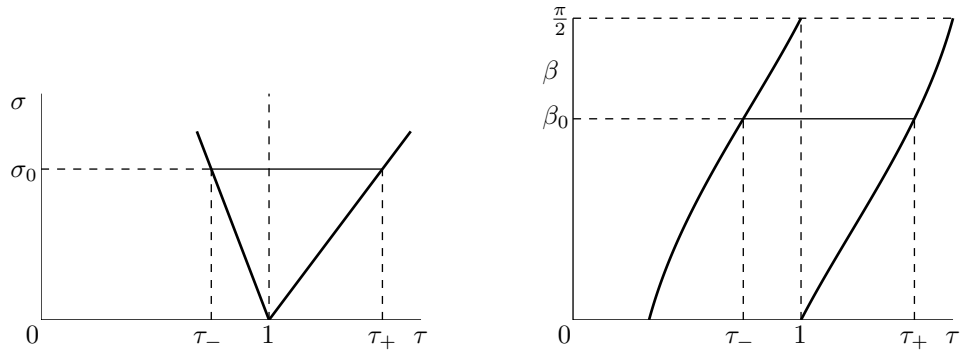}}
\put(16,33){$\sigma$}
\put(16,24){$\sigma_0$}
\put(69, 2){$\tau$}
\put(40, 2){$\tau_-$}
\put(63, 2){$\tau_+$}
\put(18, 2){0}
\put(49, 2){1}
\put(87, 44){$\frac{\pi}{2}$}
\put(86, 37){$\beta$}
\put(86, 31){$\beta_0$}
\put(139, 2){$\tau$}
\put(110, 2){$\tau_-$}
\put(133, 2){$\tau_+$}
\put(88,  2){0}
\put(119, 2){1}
\end{picture}
\caption{\textit{Left: the main tongue of the circle map $f_{\nu}$, where $\nu=(\sigma,\beta,\alpha,\tau)$ with $\alpha=0.3$ and $\beta=\beta_0$ fixed. For $\sigma=\sigma_0$ the range of entrainment is the width of the tongue at $\sigma=\sigma_0$, namely the interval $(\tau_-,\tau_+)$. Right: the range of entrainment $(\tau_-,\tau_+)$ for fixed $\sigma=\sigma_0$ and $\beta$ varying in $(0,\frac{\pi}{2})$. The parameter $\beta$ determines the ratio of $\eps$ and $\eta$ since $\eps=\sigma\cos\beta$ and $\eta=\sigma\sin\beta$.}\label{fig:range}}
\end{figure}

In section \ref{sec:fixpo} on fixed points of the map $\fmu$ we noted that the position of stable fixed points is an increasing function of the intrinsic period $\tau$. This has been observed by various authors, for example \cite{asc65,rm03} and for particular organisms by \cite{asc81,ap78,pd76,rm01}. In figure \ref{fig:relonsets} we show data from \cite{rm01} essentially giving the relative onset times of the activity interval as a function of the intrinsic period, of several mutants of the fungus \emph{Neurospora crassa}.

\begin{figure}[htbp]
\setlength{\unitlength}{1mm}
\begin{picture}(100,70)(0,0)
\put(30,5){\includegraphics[width=10cm]{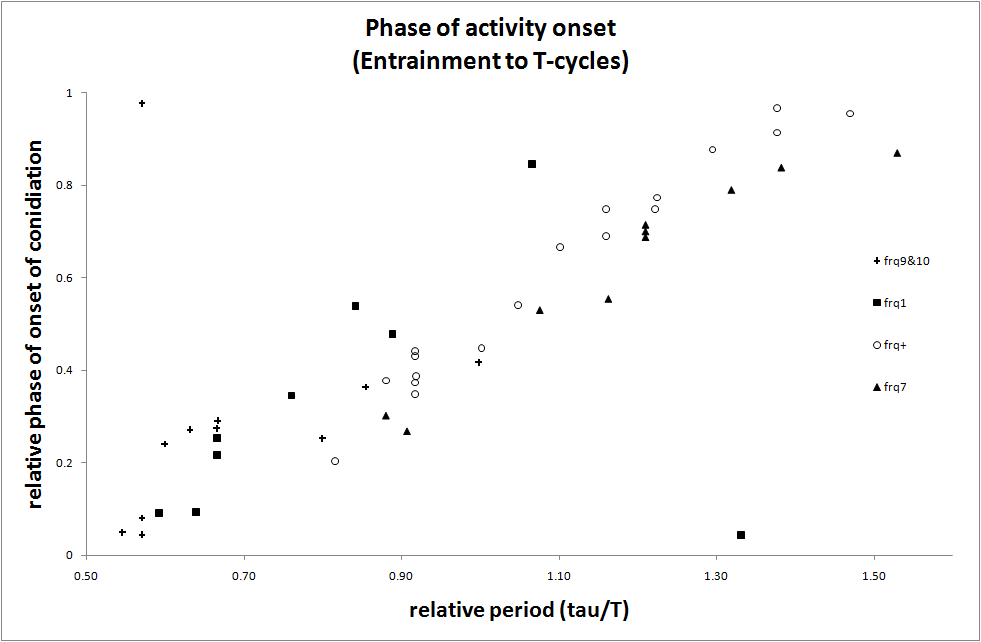}}
\end{picture}
\caption{\textit{Relative times of onset of activity interval for four mutants of \emph{Neurospora crassa}. The genetic types are indicated by markers, see \cite{rm01} for an explanation of \emph{frq1}, \emph{frq+}, \emph{frq7} and \emph{frq9\&10}.}\label{fig:relonsets}}
\end{figure}

Both position and length of the interval $(\tau_-,\tau_+)$ depend on $\alpha$ and $\beta$. As we see from figure \ref{fig:range} the range of entrainment is in general not centered at $\tau=1$. Here $\beta$ is the most important parameter. If $\beta < \frac{\pi}{4}$ or equivalently $\eps > \eta$, then phase delay is larger than phase advance and the range of entrainment is shifted towards intrinsic periods $\tau$ smaller than the period $T$ of the Zeitgeber. If $\beta > \frac{\pi}{4}$, the range of entrainment is shifted in the direction of intrinsic periods $\tau$ larger than the period $T$ of the Zeitgeber. The model has two extreme cases, one for $\beta=0$ or equivalently $\eta=0$ (only phase delay), where only pacer cells with intrinsic period $\tau$ less than the period $T$ of the Zeitgeber can be entrained. The other one is for $\beta=\frac{\pi}{2}$ or equivalently $\eps=0$ (only phase advance), where $\tau$ must be larger than $T$ for entrainment to occur. A phenomenon related to this skewness has been observed in several nocturnal rodents \cite{pd76} where a decreasing center of the range of entrainment corresponds to increasing phase delay and decreasing phase advance.

As mentioned before, in the main tongue, the map $\fmu$ has at least a pair of fixed points, one stable and one unstable. However more pairs may exist, leading to curves of saddle-node bifurcations inside the main tongue, see figure \ref{fig:snbifs}. Thus upon varying $\tau$ a fixed point may lose its stability and the system jumps to another stable fixed point. As long as parameter values remain in the main tongue this is the only possibility. Another possibility is to keep the parameters $\mu$ fixed but vary a parameter in the Zeitgeber. For a biological example possibly related to such a mechanism see \cite{spo}.

\subsection*{Synchronization}\label{sec:dcsyn}
Apart from entrainment, the model also shows the possibility of the more general phenomenon called synchronization. This occurs for parameter values in the $p:q$ tongues in the $(\tau,\sigma)$-plane. Then the circle map $\fmu$ has a $q$-periodic orbit consisting of points $t_0,t_1,\ldots,t_{q-1}$ indicating the beginnings of $q$ activity intervals of the pacer cell in $p$ periods of the Zeitgeber. In the $p:q$ tongues we have the same phenomena as in the main tongue, the only difference is that they apply to periodic points instead of fixed points.\commentaar{bio refs, examples, certain nocturnal animals with two activity intervals in one 24-hour period? 1:2}

There is an example of 2:1 periodic point, the fungus \emph{Neurospora sp} has one activity interval in two periods of the Zeitgeber, see \cite{mbrmgr}. Here the Zeitgeber is a temperature stimulus.

\subsection*{No synchronization}\label{sec:dcqp}
For parameter values outside the tongues the circle map has quasi-periodic orbits. This corresponds to quasi-periodic occurrence of activity intervals of the pacer cell. In practice such behavior may be hard to distinguish from periodic behavior with a long period.\commentaar{bio refs, examples}

%

\subsection*{Future directions}\label{sec:dcfudi}
There are several ways to generalize or extend the current model for a single pacer cell with a periodic Zeitgeber. We first concentrate on the biological model of section \ref{sec:biomod} in view of our future goal to describe a collection of interacting pacer cells, stimulated by a Zeitgeber. 

\paragraph{Single pacer cell, circle map} In the present analysis we restricted to the case that the map $\fmu$ is a circle diffeomorphism (invertible map). Then the $(\tau,\sigma)$-plane is divided into tongues with well-defined periodic dynamics and hairs with quasi-periodic dynamics. This only occurs when the stimulus of the Zeitgeber on the pacer cell is relatively weak. Allowing a stronger stimulus, the map $\fmu$ becomes an endomorphism (non-invertible map) with far richer bifurcation scenarios, see for example \cite{bst}. However, it is not clear whether the subtleties of the endomorphism case are in accordance with the coarseness of the underlying biological model. Therefore we restrict to the simplest properties of each model.

\paragraph{Single pacer cell, torus map} Thus it seems more fruitful to generalize in another direction and consider a quasi-periodic Zeitgeber. This will lead to a torus map instead of a circle map. Such a generalization is also more relevant for our aim to study a collection of pacer cells, by first considering a single cell with an asymmetric interaction with its environment. The latter stimulates the pacer cell, but not vice versa.

\paragraph{Collection of pacer cells, no interaction} In a model for a collection of pacer cells we can already use the results for a single cell. As a first model let us assume that the cells are stimulated by an external periodic Zeitgeber but do not interact. Although biologically not particularly relevant, it is a step in gradually sophisticating the model. Furthermore suppose that there are $n$ cells, characterized by parameter values $(\eps_i, \eta_i, \alpha_i, \tau_i)$ in the main tongue for $i=1,\ldots,n$. That is we apply the current single pacer cell model for each cell. Then the collection will be entrained to the Zeitgeber albeit that each cell has its own onset time and length of activity interval. The collective behavior though, will be periodic.

\paragraph{Collection of pacer cells, with interaction} We conjecture that, starting with the previous model, for a sufficiently weak interaction there will still be entrainment. Nevertheless it will be interesting to consider a model for interacting pacer cells without Zeitgeber as well. Here we may take the pacer cells nearly identical in the sense of the previous paragraph. However there is biological evidence that there are different types of pacer cells \cite{hnsh,wahh,wlmr}. In further extensions one could again include a periodic Zeitgeber. Since in a model for a collection of pacer cells the Zeitgeber acts solely as an external stimulus, it seems most natural to restrict to periodic Zeitgebers. However we may wish to include both daily and seasonally variations. This would imply that the period of the Zeitgeber is a year. Another possibility is to stick to Zeitgebers with a 24-hour period and use methods for slowly varying parameters to include seasonal changes.


%

\appendix

\section{Bifurcations}\label{sec:apppfbifs}

\commentaar{some introductory talk on:
- pf bifurcation
- parameter dependent Zeitgeber
- period doubling in tongues for endomorphisms, referring to BST
- period doubling picture of Kim
}

The tongues in the $(\sigma,\tau)$-parameter plane are determined by saddle-node bifurcations. But there may also be other bifurcations even when the map $\fmu$ is a diffeomorphism, in other words invertible. The reason is that we have many parameters. Considering $\alpha$ as a relatively unaccessible parameter and keeping it fixed we may still vary $\sigma$, $\tau$ and $\beta$. Then we have 3-dimensional tongues in $(\sigma,\tau,\beta)$-parameter space. It is an almost straightforward consequence of theorem \ref{pro:fixpos} that there are curves of pitchfork bifurcations in this three dimensional parameter space, emanating from rational points on the $\tau$-axis.

\begin{corollary}[Pitchfork bifurcation]\label{col:pfbif}
Let $\zgbr$ be a 1-periodic Zeitgeber. If $(\zgbr'')^2 - \zgbr'\cdot\zgbr'''$ has a simple zero, then parameter values exist for which $\fmu$ has a pitchfork bifurcation.
\end{corollary}

\begin{proof}
From the proof of proposition \ref{pro:fixpos} we see that the number of solution branches of equation \eqref{eq:tster} does change if $\zgbr''(t)^2 - \zgbr'(t) \cdot \zgbr'''(t)$ has a simple zero.
\end{proof}

However, $\beta$ may be considered as an unaccessible parameter as well. But the Zeitgeber may also depend on a parameter. We may in particular view seasonal change, which is slow compared to the 24-hour period, as parameter dependence. The quantity $(\zgbr'')^2 - \zgbr'\cdot\zgbr'''$ may change sign depending on this parameter, so we find again pitchfork bifurcations.\commentaar{bio implications?}

If the map $\fmu$ is not a diffeomorphism (not invertible) then there are numerous other bifurcations, see \cite{bst}. This happens for relatively large values of $\sigma$. On varying $\tau$ for fixed values of the other parameters in the main tongue, one generally finds a number of period doublings followed by the same number of period halvings (or in opposite order). The reason is that there are curves of period doublings in the $(\tau,\sigma)$-parameter plane with local minima, considered as functions of $\tau$, that are transversally crossed. For more details we refer again to \cite{bst}. An example of this phenomenon is shown in figure \ref{fig:pdh}.

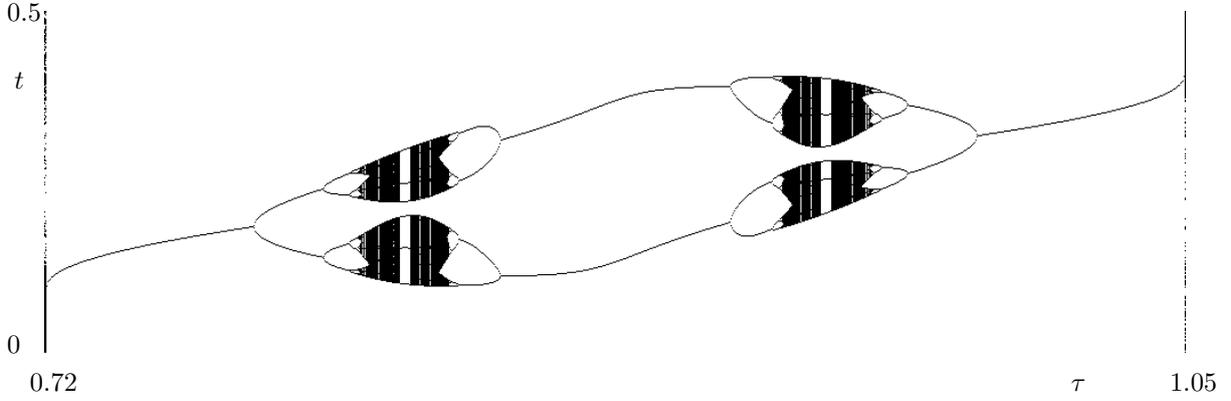
\begin{figure}[htbp]
\setlength{\unitlength}{1mm}
\begin{picture}(100,50)(0,0)
\put(5,5){\includegraphics[width=15cm]{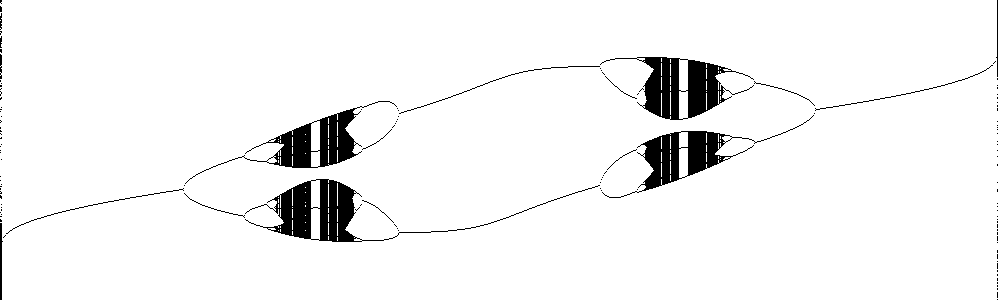}}
\put(140,0){$\tau$}
\put(1, 40){$t$}
\put(0,  5){$0$}
\put(0, 49){$0.5$}
\put(3,  0){$0.72$}
\put(153,0){$1.05$}
\end{picture}
\caption{\textit{Period doublings and halvings for the map $\fmu$ with Zeitgeber $Z(t) = 3+\sin(2\pi t)+2\cos(4\pi t)$. Shown are the positions of periodic points as a function of the intrinsic period $\tau$. Parameter values are $\alpha=0.5$, $\eps=0.4$ and $\eta=0.17$.}\label{fig:pdh}}
\end{figure}

\section{Proofs}\label{sec:appproofs}
\begin{proof}[Proof of proposition \ref{pro:ds}]
The main point we have to show is that equation \eqref{eq:t} can be solved uniquely with respect to $t$. Using the expression for $\theta$ in \eqref{eq:theta} we obtain after some rearranging
\begin{equation}\label{eq:tnp1}
t_{n+1} + \eta \zgbr(t_{n+1}) = t_n + \alpha + \eps \zgbr(t_n+\alpha) - \alpha + \tau.
\end{equation}
Since $\zgbr$ is $1$-periodic, $U_{\eps}$ has the property $U_{\eps}(t+1) = U_{\eps}(t) + 1$, for all $t$. Then the equation for $t_{n+1}$ reads
\begin{equation*}
U_{\eta}(t_{n+1}) = U_{\eps}(t_n+\alpha) - \alpha + \tau.
\end{equation*}
\commentaar{wat voegt die operator toe??}
Introducing the operator $\cT_{\alpha}$ which takes a function $f$ into $\cT_{\alpha} f = T_{-\alpha} \circ f \circ T_{\alpha}$ where $T_{\alpha}$ is just translation over $\alpha$, that is $T_{\alpha}(t) = t+\alpha$ we can write the equation for $t_{n+1}$ as
\begin{equation}\label{eq:tnplus1}
U_{\eta}(t_{n+1}) = \cT_{\alpha} U_{\eps}(t_n) + \tau.
\end{equation}
Solvability of \eqref{eq:tnplus1} depends on the value of $\eta$. If $\eta$ is small enough, $U_{\eta}$ is invertible and we write
\begin{equation*}
t_{n+1} = \Fmu(t_n) = U_{\eta}^{-1}(\cT_{\alpha} U_{\eps}(t_n) + \tau).
\end{equation*}
Then $\Fmu$ is a differentiable map with the property $\Fmu(t+1)=\Fmu(t)+1$, which means that $\Fmu$ is the \emph{lift} of a \emph{circle map} of degree one. Note that $\Fmu$ depends on the parameters $\mu = (\eps, \eta, \alpha, \tau)$.
\end{proof}

\begin{proof}[Proof of theorem \ref{the:standard}]
A consequence of $\Fmu(t+1)=\Fmu(t)+1$ is that $\Fmu(t)-t$ is 1-periodic, thus there is a 1-periodic $\cinf$ function $P_{\nu}$ such that $\Fmu(t)=t+P_{\nu}(t)$. $P_{\nu}$ has a Fourier series so we may split off the constant term and we write $P_{\nu}(t)=\omega(\nu)+Q_{\nu}(t)$ with $\omega(\nu)=\int_0^1 P_{\nu}(t)\,dt$, then $\omega$ is a $\cinf$ function of $\nu$. Furthermore $Q_{\nu}(t) = P_{\nu}(t) - \omega(\nu)$ so that $Q_{\nu}$ is a 1-periodic $\cinf$ function with $\int_0^1 Q_{\nu}(t)\,dt=0$. So far $\Fmu(t)=t+\omega(\nu)+Q_{\nu}(t)$. For $\sigma=0$ we have $\Fmu(t)=t+\tau$, so $\omega(0,\beta,\alpha,\tau)=\tau$ and $Q_{(0,\beta,\alpha,\tau)}(t)=0$. Using the division property of $\cinf$ functions, a 1-periodic, $\cinf$ function $R_{\nu}$ exists such that $Q_{\nu}(t)=\sigma R_{\nu}(t)$. Finally $F_{(\sigma \cos\,\beta, \sigma \sin\,\beta, \alpha, \tau)}(t)=t+\omega(\nu)+\sigma R_{\nu}(t)$.
\end{proof}

\begin{proof}[Proof of proposition \ref{pro:bmt}]
Recall that $\mu=(\eps,\eta,\alpha,\tau)$ and $\nu=(\sigma,\beta,\alpha,\tau)$ with $\eps=\sigma\cos\beta$ and $\eta=\sigma\sin\beta$. Then for small $\sigma$ and assuming that $1-\tau = \cO(\sigma)$ we get
\begin{align*}
F_{\mu(\nu)} &= U_{-\eta}(U_{\eps}(t+\alpha)-\alpha+\tau) + \cO(\sigma^2)\\
            &= t + \tau + \eps Z(t+\alpha) - \eta Z(t+\tau) + \cO(\sigma^2)\\
            &= t + \tau + \eps Z(t+\alpha) - \eta Z(t) + \cO(\sigma^2).
\end{align*}
The Zeitgeber $\zgbr$ has the following Fourier series $\zgbr(t) = c_0 + c_1 \sin(2\pi t) +\sum_{k>1} c_k \sin(2\pi(kt+\gamma_k)$. After a near identity transformation followed by a time shift we obtain that $F_{\mu(\nu)}$ is equivalent to
\begin{equation*}
G_{\mu(\nu)}(t) = t + \tau + (\eps-\eta)c_0 + \eps c_1 \sin(2\pi(t+\alpha)) - \eta c_1 \sin(2\pi t) + \cO(\sigma^2)
\end{equation*}
It almost immediately follows that the boundaries of the main tongue of $G$ and thus of $F$ are as stated in the lemma.
\end{proof}

\begin{proof}[Proof of proposition \ref{pro:fixpos}]
The number of fixed points may change if the number of solution branches of equation \eqref{eq:tster}: $U_{\eps}(t+\alpha)-\alpha+\tau = U_{\eta}(t+1)$ changes. This is equivalent to a changing number of extrema of $\tau = U_{\eta}(t) - U_{\eps}(t+\alpha) + \alpha + 1$. Using $U_{\eps}(t) = t + \eps \zgbr(t)$ we get $\tau = \eta\zgbr(t) -\eps\zgbr(t+\alpha) + \alpha +1$. Thus the number of extrema of $\tau$ changes at parameter values for which
\begin{equation*}
\left\{\begin{array}{r}
\eta\zgbr'(t) -\eps\zgbr'(t+\alpha) = 0\\
\eta\zgbr''(t) -\eps\zgbr''(t+\alpha) = 0
\end{array}\right.
.
\end{equation*}
This equation has trivial solutions for $\eps$ and $\eta$ only if $\zgbr'(t) \cdot \zgbr''(t+\alpha) - \zgbr'(t+\alpha) \cdot \zgbr''(t) \neq 0$. We rewrite this as
\begin{equation*}
\frac{\zgbr'(t)}{\zgbr''(t)} \neq \frac{\zgbr'(t+\alpha)}{\zgbr''(t+\alpha)}.
\end{equation*}
This inequality holds if $h(t)=\frac{\zgbr'(t)}{\zgbr''(t)}$ is injective. Therefore a sufficient condition is that $h'$ does not change sign. From
\begin{equation*}
h'(t) = \frac{\zgbr''(t)^2 - \zgbr'(t) \cdot \zgbr'''(t)}{\zgbr''(t)^2}
\end{equation*}
the result follows.
\end{proof}


\end{document}